\documentclass[12pt,reqno]{amsart}

\usepackage{epsfig}
\usepackage{amsmath, amsthm, amsfonts}
\usepackage{graphicx}
\usepackage{amssymb}
\DeclareMathOperator {\esssup}{ess\,sup}
\DeclareMathOperator {\essinf}{ess\,inf}

\newcommand{\ga}{\gamma}
\newcommand{\Ga}{\Gamma}
\newcommand{\La}{\Lambda}

\newcommand{\ep}{\varepsilon}
\newcommand{\de}{\delta}

\newcommand{\f}{\phi}
\newcommand{\sg}{\sigma}

\newcommand{\ta}{\theta}

\newtheorem{theo}{{\sc \bf Theorem}}[section]

\newtheorem{lem}[theo]{{\sc \bf Lemma}}

\numberwithin{equation}{section}
\begin{document}

\title[Cloud computing and hyperbolic Voronoi diagrams on the sphere]
{Cloud computing and hyperbolic Voronoi diagrams on the sphere}
\pagestyle{plain}

\author{Pavel Bleher}
\address{Department of Mathematical Sciences,
Indiana University-Purdue University Indianapolis,
402 N. Blackford St., Indianapolis, IN 46202, U.S.A.}
\email{bleher@math.iupui.edu}

\author{Caroline Shouraboura}
\address{Forest Ridge \@
              9927 171st Ave SE, Newcastle, WA 98059 }
\email{carolinesh@forestridge.org}

\thanks{The first author is supported in part
by the National Science Foundation (NSF) Grant DMS-0969254.}

\date{\today}

\begin{abstract} 
In this work we study the minimization problem for the total distance in a cloud computing 
network on the sphere. We prove that  a solution to this problem is given in terms of a hyperbolic 
Voronoi diagram on the sphere. We present results of computer simulations illustrating
the solution.
\end{abstract}

\maketitle
%%%%%%%%%%%%%%%%%%%%%%%%%%%%%%%%%%%%%%%%%%%%%%%%%%%%%%%%%%%%%%%%%%%%%%%%%%%%%%%%%%%%%%%%%%%%%%%%%%
%%%%%%%%%%%%%%%%%%%%%%%%%%%%%%%%%%%%%%%%%%%%%%%%%%%%%%%%%%%%%%%%%%%%%%%%%%%%%%%%%%%%%%%%%%%%%%%%%%
%%%%%%%%%%%%%%%%%%%%%%%%%%%%%%%%%%%%%%%%%%%%%%%%%%%%%%%%%%%%%%%%%%%%%%%%%%%%%%%%%%%%%%%%%%%%%%%%%%
%%%%%%%%%%%%%%%%%%%%%%%%%%%%         INTRODUCTION              %%%%%%%%%%%%%%%%%%%%%%%%%%%%%%%%%%%%%
%%%%%%%%%%%%%%%%%%%%%%%%%%%%%%%%%%%%%%%%%%%%%%%%%%%%%%%%%%%%%%%%%%%%%%%%%%%%%%%%%%%%%%%%%%%%%%%%%%
%%%%%%%%%%%%%%%%%%%%%%%%%%%%%%%%%%%%%%%%%%%%%%%%%%%%%%%%%%%%%%%%%%%%%%%%%%%%%%%%%%%%%%%%%%%%%%%%%%

\section{Introduction}

We consider a network 
\[
\mathcal S=\{S_1,\ldots,S_n\}
\]
 of computer data centers on the earth surface. 
We assume that the earth surface is represented by a sphere $\La_R$ of a radius $R>0$, and each data 
center $S_m$ is described by its geographical coordinates, the latitude $\f_m$ and
the longitude $\ta_m$. The data centers provide a variety of
cloud computing services to a big set of users 
\[
\mathcal U=\{U_1,\ldots,U_N\}.
\]
 Our main goal will be to solve the {\it minimization problem for the total distance} $\mathcal D$ between the data 
centers and the users. We assume that each user $U_k$ is assigned to a data center $S_{j(k)}$, 
and the total distance is the sum of the distances between the users and the 
corresponding data centers,
\begin{equation}\label{in1}
\mathcal D(j)=\sum_{k=1}^N d(S_{j(k)},U_k),
\end{equation}
where $j=j(k)$ is the function assigning the user $U_k$ to the data center $S_{j(k)}$.

For any two points $x,y$ on the sphere we denote $d(x,y)$ the distance between $x$ and $y$ on the sphere,
that is the length of the arc of the great circle connecting $x$ to $y$.
We will consider the problem of the minimization of the total distance under certain constraints.
We will assume that each data center $S_j$ has a capacity $C_j$ which characterizes how many
users it can service. 
The assumption of the limited capacity for the data centers leads to the minimization
problem with constraints: Find an assignment function $j_0(k)$ such that   
\begin{equation}\label{in5}
\min_{j\in \mathcal J(C)} \mathcal D(j)=\mathcal D(j_0),
\end{equation}
where 
\begin{equation}\label{in6}
\mathcal J(C)=\big\{j=j(k)\;\big|\;N_m=\#\{k:\; j(k)=m\}\le C_m\quad \textrm{for}\quad m=1,\ldots,n\big\}. 
\end{equation}
Observe that $N_m=\#\{k:\; j(k)=m\}$ is just the number of users assigned to the data center $S_m$. 

Since the number of users is usually large, it can be useful to consider a 
(nonnegative) measure $\mu$ 
of users on the sphere. In this case, the total distance functional looks like
\begin{equation}\label{in7}
\mathcal D(j)=\int_{\La_R} d(S_{j(U)},U)\,d\mu(U),
\end{equation}
where $j(U)$ is the assignment function of the users to the data centers.
Respectively, constraint \eqref{in6} takes the form
\begin{equation}\label{in8}
\mathcal J(C)=\left\{j=j(U)\;\big|\;\mu \{U:\, j(U)=m\}\le C_m\quad \textrm{for}\quad m=1,\ldots,n\right\}. 
\end{equation}

The main result of this work is Theorem 2.1, in which we prove that the solution to minimization problem 
\eqref{in5} is given by a hyperbolic Voronoi diagram on the sphere. 
The classical Voronoi diagrams have numerous applications to science and technology (for many of 
these applications see the review paper \cite{aur} by
Aurenhammer, the lectures \cite{mol} by M\o ller, the collection of papers \cite{Gav}, the monograph \cite{OBS} of Okabe, 
Boots, Sugihara, and  Chiu, and references therein). For applications of Voronoi diagrams 
to cloud computing see the recent paper \cite{SB} of Shouraboura and Bleher and references therein. In our
Theorem 2.1 we show that the presence of the constraints leads to an extension of the classical Voronoi
diagrams to hyperbolic Voronoi diagrams on the sphere. It is interesting to notice that 
in a completely different setting  similar diagrams on the plane or in the 3D space appear 
in crystallography as the Johnson-Mehl model of the crystal growth \cite{JM}. We formulate and prove Theorem 2.1 in Section \ref{CMP}. In the subsequent section we illustrate Theorem 2.1 by results of computer simulations.   
%%%%%%%%%%%%%%%%%%%%%%%%%%%%%%%%%%%%%%%%%%%%%%%%%%%%%%%%%%%%%%%%%%%%%%%%%%%%%%%%%%%%%%%%%%%%%%%%%%
%%%%%%%%%%%%%%%%%%%%%%%%%%%%%%%%%%%%%%%%%%%%%%%%%%%%%%%%%%%%%%%%%%%%%%%%%%%%%%%%%%%%%%%%%%%%%%%%%%
%%%%%%%%%%%%%%%%%%%%%%%%%%%%%%%%%%%%%%%%%%%%%%%%%%%%%%%%%%%%%%%%%%%%%%%%%%%%%%%%%%%%%%%%%%%%%%%%%%
%%%%%%%%%%%%%%%%%%%%%%%%%%%%         SOLUTION              %%%%%%%%%%%%%%%%%%%%%%%%%%%%%%%%%%%%%
%%%%%%%%%%%%%%%%%%%%%%%%%%%%%%%%%%%%%%%%%%%%%%%%%%%%%%%%%%%%%%%%%%%%%%%%%%%%%%%%%%%%%%%%%%%%%%%%%%
%%%%%%%%%%%%%%%%%%%%%%%%%%%%%%%%%%%%%%%%%%%%%%%%%%%%%%%%%%%%%%%%%%%%%%%%%%%%%%%%%%%%%%%%%%%%%%%%%%

\section{Solution of the constrained minimization problem and  
hyperbolic Voronoi diagrams on the sphere }\label{CMP}

In this section we will solve minimization problem \eqref{in5}. We will consider a general
measure $\mu$ of users. For a given measurable assignment function $j(U)$ on $\La_R$, we introduce the sets  
\begin{equation}\label{cm1}
\sg_m(j)=\big\{U\in\La_R:\; j(U)=m\big\},\qquad  m=1,\ldots,n,
\end{equation}
so that $\sg_m(j)$ is the set of users assigned to the data center $S_m$.
The sets $\{\sg_m(j)\}$ are disjoint and they form a measurable partition $\pi(j)$
of the sphere $\La_R$, so that
\begin{equation}\label{cm2}
\pi(j):\;\La_R=\bigsqcup_{m=1}^n \sg_m(j).
\end{equation}
The functional $\mathcal D(j)$ reduces to
\begin{equation}\label{cm3}
\mathcal D(j)=\sum_{m=1}^n \int_{\sg_m(j)} d(S_{m},U)\,d\mu(U),
\end{equation}
and constraint \eqref{in8} to
\begin{equation}\label{cm4}
\mathcal J(C)=\left\{j=j(U)\;\big|\;\mu \left(\sg_m(j)\right)\le C_m\quad \textrm{for}\quad m=1,\ldots,n\right\}. 
\end{equation}

Let us begin with the unconstrained minimization problem. 
In the absence of the capacity constraint, minimization 
problem \eqref{in5} can be solved by assigning each user $U$ to its closest data center $S_{j}$,
so that $S_{j_0(U)}$ is the closest data center to $U$.
Geometrically, we partition the sphere $\La_R$ into the cells $\{\sg_m\}$ of the Voronoi diagram 
$V$ on the sphere with the points $\{S_m\}$.
The cell $\sg_m$ is defined as the set of points on $\La_R$ which are closer (or at the same distance) to $S_m$ 
than to other $S_l$'s, that is
\begin{equation}\label{cmp1}
\sg_m=\big\{x\in\La_R:\; d(x,S_m)\le d(x,S_l)\;\textrm{for all}\; l\not=m\big\}.
\end{equation}
Observe that the cells $\{\sg_m\}$ are convex spherical polygons on the sphere $\La_R$. With the Voronoi diagram
we associate a graph $\Ga_V$.
The vertices  of the graph $\Ga_V$ are the vertices of the polygons $\{\sg_m\}$, 
and the edges of the graph $\Ga_V$ are the sides of the polygons $\{\sg_m\}$. 
The Delaunay triangulation, associated with the Voronoi diagram, is the dual graph $\Ga_D$ to
$\Ga_V$, with vertices $\{S_m\}$ and
edges connecting vertices $S_m$, $S_l$ if and only if the cells $\sg_m$, $\sg_l$ have a common side.

Thus, in the absence of the constraint, we assign to each data center $S_{m}$ all users in the
cell $\sg_m$ of the Voronoi diagram. To describe the assignment in the presence of the constraint we
introduce hyperbolic Voronoi diagrams on the sphere.

Suppose that to each point $S_m$ a number $d_m$ is assigned. Then the {\it hyperbolic Voronoi diagram}
\[
V(d), \quad d=(d_1,\ldots, d_n),
\]
with the parameters $\{d_m\}$ and the points $\{S_m\}$, is defined as follows.  
The cell $\sg_m$ of the point $S_m$ is defined as 
\begin{equation}\label{cmp2}
\sg_m=\left\{x\in\La_R:\; d(x,S_m)+d_m\le d(x,S_l)+d_l\;\textrm{for all}\; l\not=m\right\}.
\end{equation}
Observe that on the curve $\ga_{ml}$ separating two neighboring cells $\sg_m$, $\sg_l$ we have the equation,
\begin{equation}\label{cmp3}
d(x,S_m)+d_m= d(x,S_l)+d_l,
\end{equation}
hence $\ga_{ml}$ is a part of the spherical hyperbola on the sphere $\La_R$. A vertex $v$ of the 
hyperbolic Voronoi diagram is a point which belong to three or more cells. The graph $\Ga_V(d)$
of the hyperbolic Voronoi diagram $V(d)$ consists of the vertices $\{v\}$ and the edges $\{e\}$,
which are the curves $\{\ga_{ml}\}$ separating neighboring cells.

A solution to the constrained minimization problem can be obtained as follows. 
We will take a natural assumption that the total capacity
of all data centers is not less than the number of the users, 
\begin{equation}\label{cmp4}
\sum_{m=1}^n C_m\ge \mu\La_R.
\end{equation}
Otherwise, it would be impossible to service all the users.

\begin{theo} \label{thm1} For any measure $\mu$, a minimizer $j_0(U)$ to constrained minimizing
problem \eqref{in5} exists, which can be obtained 
as follows.
There exist numbers $(d_1,\ldots, d_n)$ such that the minimizer $j_0$
 is obtained by assigning all users in the cell $\sg_m$ of 
the hyperbolic Voronoi diagram $V(d)$ to the data center $S_m$,
so that
\begin{equation}\label{cmp5}
j_0(U)=m \quad \textit{if and only if}\quad U\in \sg_m.
\end{equation}
\end{theo} 

\begin{proof} For any $\ep>0$, we will introduce a regularization $\mu_\ep(U)$ of 
the measure $\mu(U)$, and we will first prove Theorem \ref{thm1} for the measure
$\mu_\ep(U)$ and then we will take the limit $\ep\to 0$. The measure $\mu_\ep(U)$
is defined as follows.
Consider the function
\begin{equation}\label{cmp6}
g_\ep(x)= e^{-x^2/\ep}, 
\end{equation}
 on the real line,
and for a given point $U_0\in\La_R$, introduce the probability density $f_\ep(U;U_0)$ on the sphere $\La_R$ as
\begin{equation}\label{cmp7}
f_{\ep}(U;U_0)=\frac{1}{Z(\ep)} \,g_\ep\left(d(U,U_0)\right)>0
\end{equation}
where 
\begin{equation}\label{cmp8}
Z(\ep)=\int_{\La_R} \,g_\ep\left(d(U,U_0)\right)\,dU,
\end{equation}
Then, as $\ep\to 0$, the function $f_{\ep}(U;U_0)$ converges to a $\de$-function at the point $U_0$.
Consider the function
\begin{equation}\label{cmp9}
p_\ep(U)=\int_{\La_R} \,f_\ep(U;U_0)d\mu(U_0)>0,
\end{equation}
and define the measure $\mu_\ep(U)$ as
\begin{equation}\label{cmp10}
d\mu_\ep(U)=p_\ep(U)dU.
\end{equation}
Then in the weak sense,
\begin{equation}\label{cmp11}
\lim_{\ep\to 0} \mu_\ep=\mu,
\end{equation}
so that for any continuous function $f(U)$,
\begin{equation}\label{cmp11a}
\lim_{\ep\to 0} \int_{\La_R} \,f(U)\,d\mu_\ep(U)=\int_{\La_R} \,f(U)\,d\mu.
\end{equation}

Consider the minimization problem
\begin{equation}\label{cmp12}
\min_{j\in \mathcal J(C)} \mathcal D_{\ep}(j)=\mathcal D_{\ep}(j_\ep),
\end{equation}
for the functional
\begin{equation}\label{cmp13}
\mathcal D_\ep(j)=\int_{\La_R} d(S_{j(U)},U)\,d\mu_\ep(U).
\end{equation}

\begin{lem} \label{lem1}A minimizing function $j_\ep$ for minimization problem \eqref{cmp12} exists.
\end{lem}

\begin{proof} Let
\begin{equation}\label{cmp14}
\inf_{j\in \mathcal J(C)} \mathcal D_{\ep}(j)=\lim_{k\to\infty}\mathcal D_{\ep}(j_\ep^k).
\end{equation}
Consider the partitions $\pi(j_\ep^k)$. By the compactness argument, there exists a weakly
converging subsequence
\begin{equation}\label{cmp15}
\lim_{r\to\infty}\pi(j_\ep^{k_r})=\pi(j_\ep).
\end{equation}
Since $d(U,U_0)$ is a continuous function of $(U,U_0)$ and $p_\ep(U)$ is
a continuous function of $U$, this implies that $j_\ep$ solves
\eqref{cmp12}.
\end{proof}

\begin{lem} \label{lem2} Let $j_\ep$ be a solution to \eqref{cmp12} and  
\begin{equation}\label{pi1}
\pi(j_\ep):\;\La_R=\bigsqcup_{m=1}^n \sg_m(j_\ep)
\end{equation}
the corresponding partition of $\La_R$.
Then for any $m\not= l$,   
\begin{equation}\label{pi2}
\underset{U\in \sg_m(j_\ep)}{\essinf}\;[ d(U,S_l)- d(U,S_m)]\ge 
\underset{U\in \sg_l(j_\ep)}{\esssup}\;[ d(U,S_l)- d(U,S_m)].
\end{equation}
\end{lem}

\begin{proof} Suppose, for the sake of contradiction, that \eqref{pi2} is false. Then 
there exist sets $A\subset \sg_l(j_\ep)$, $B\subset \sg_m(j_\ep)$  and numbers $a>b$ such that
\begin{enumerate}
\item $ d(U,S_l)- d(U,S_m)>a$ for all $U\in A$,
\item $ d(U,S_l)- d(U,S_m)<b$ for all $U\in B$,
\item $\mu_\ep(A)=\mu_\ep(B)>0$.
\end{enumerate}
Let us reassign all users $U\in A$ from $S_l$ to $S_m$ and all users $U\in B$ 
from $S_m$ to $S_l$. This reassignment, $j_\ep^R$,
will not violate the constraints and it will decrease the total distance functional $\mathcal D(j)$.
Indeed,
\begin{equation}\label{cmp13a}
\begin{aligned}
\mathcal D_\ep(j_\ep)-\mathcal D_\ep(j_\ep^R)&=\int_{A} [d(U,S_l)- d(U,S_m)]\,d\mu_\ep(U)
-\int_{B} [d(U,S_l)- d(U,S_m)]\,d\mu_\ep(U)\\
&\ge a\mu_\ep(A)-b\mu_\ep(B)=(a-b)\mu_\ep(A)>0,
\end{aligned}
\end{equation}
hence $\mathcal D(j_\ep)$ is not the minimum of $\mathcal D(j)$. This contradiction proves \eqref{pi2}.
\end{proof}

It follows from \eqref{pi2} that there exist sets $\{E_m,\;m=1,\ldots,n\}$ such that 
\begin{enumerate}
\item $ E_m\subset \sg_m(j_\ep)$ for  $m=1,\ldots,n$,
\item $\mu_\ep \left(\sg_m(j_\ep)\setminus E_m\right)=0$ for $m=1,\ldots,n$,
\item for any $m\not= l$,
\begin{equation}\label{pi3}
\underset{U\in E_m}{\inf}\;[ d(U,S_l)- d(U,S_m)]\ge 
\underset{U\in E_l}{\sup}\;[ d(U,S_l)- d(U,S_m).
\end{equation}
\end{enumerate}

Consider the closure $F_m$ of each set $E_m$, 
\begin{equation}\label{pi4}
F_m=\overline{E_m},\qquad m=1,\ldots,n.
\end{equation}
Then, since $\{\sg_m(j_\ep)\}$ is a partition of the sphere $\La_R$, we have that
\begin{equation}\label{pi5}
\La_R=\bigcup_{m=1}^n F_m.
\end{equation}
By continuity, \eqref{pi3} holds if we replace $E_m,E_l$ for $F_m,F_l$, 
respectively, that is for any $m\not= l$,
\begin{equation}\label{pi6}
\underset{U\in F_m}{\inf}\;[ d(U,S_l)- d(U,S_m)]\ge 
\underset{U\in F_l}{\sup}\;[ d(U,S_l)- d(U,S_m).
\end{equation}
Moreover, if the intersection of $F_m$ and $F_l$ is nonempty,
\begin{equation}\label{pi7}
\ga_{ml}\equiv F_m\cup F_l\not=\emptyset,
\end{equation}
then
\begin{equation}\label{pi8}
\underset{U\in F_m}{\inf}\;[ d(U,S_l)- d(U,S_m)]= 
\underset{U\in F_l}{\sup}\;[ d(U,S_l)- d(U,S_m)\equiv d_{ml}.
\end{equation}
This implies that for $U\in\ga_{ml}$,
\begin{equation}\label{pi9}
 d(U,S_l)- d(U,S_m)= d_{ml},
\end{equation}
Thus, $\ga_{ml}$ is a part of  hyperbola \eqref{pi9} on the sphere $\La_R$. 
Since the function $d(U,S_m)$ is an analytic function of $U$ outside of the point $S_m$,
it follows that $\ga_{ml}$ consists of a finite number of arcs of  hyperbola \eqref{pi9}.  
Moreover, the boundary of $F_m$ is equal to
\begin{equation}\label{pi10}
\partial F_m=\bigcup _{l\not=m} (F_m\cap F_l),
\end{equation}
hence $\partial F_m$ consists of a finite number of arcs of hyperbolas on the sphere $\La_R$.

Assume that for some $m$ the sets $\ga_{ml}$ and $\ga_{mp}$ intersect at some point $v$. Then $v$ belongs
to the boundaries of $F_l$ and $F_p$, hence 
$\ga_{lp}$ passes through $v$, that is $v$ is a triple point (or higher). By taking $U=v$ in \eqref{pi9},
we obtain that 
\begin{equation}\label{pi11}
d_{ml}+d_{lp}+d_{pm}=0.
\end{equation}
Define the one-chain 
\begin{equation}\label{pi12}
\mathcal C_1=\sum_{m,l}d_{ml}\ga_{ml}.
\end{equation}
Then by \eqref{pi11},
\begin{equation}\label{pi13}
\partial \mathcal C_1=0.
\end{equation}
Since the first homology group of the sphere $\La_R$ is trivial, this implies that there exists 
a zero-chain $\mathcal C_0$ such that 
\begin{equation}\label{pi14}
\partial \mathcal C_0=\mathcal C_1.
\end{equation}
In other words,  there exist numbers $\{d_m\}$
such that 
\begin{equation}\label{pi15}
d_{ml}=d_m-d_l.
\end{equation}
By \eqref{pi9}, this proves \eqref{cmp3} and hence Theorem \ref{thm1} for the distribution $\mu_\ep$ of users. 
Let us consider the distribution $\mu$.

The numbers $\{d_m\}$ are defined uniquely up to a common additive constant. To fix the constant, let us put $d_1=0$.
Then the other numbers, $d_m=d_m(\ep)$, $m>1$, are bounded by a constant independent of $\ep$. Let us take a converging 
subsequence 
\begin{equation}\label{pi16}
d_m=\lim_{\ep_p\to 0} d_m(\ep_p),\quad p=1,2,\ldots;\qquad m=1,\ldots,n.
\end{equation}
This defines the numbers $\{d_m\}$. Consider the corresponding hyperbolic Voronoi diagram $V(d)$.
Since $V(d(\ep))$ gives a solution to the minimization problem for the users distribution $\mu_\ep$,
we obtain that $V(d)$ gives a solution for the distribution $\mu$. 
This proves Theorem \ref{thm1}.   
\end{proof}

{\it Remark.} Observe that, in general, the hyperbolic Voronoi diagram solving the minimization problem {\it is not unique}.
Consider, for instance, a discrete distribution of users
\begin{equation}\label{pi17}
d\mu(U)=\sum_{k=1}^N \de(U-U_k)\,dU.
\end{equation}
Then, if there are no users on the edges of the Voronoi diagram, then we can shift the numbers 
$d_1,\ldots,d_n$ a little, without changing the number of users in the Voronoi cells.
On the other hand, the  solution to the minimization problem {\it is unique} for any users distribution which
has an absolutely continuous component with a positive density, like $\mu_\ep$. This explains why
we used the distribution $\mu_\ep$ in the above proof, establishing the theorem first for $\mu_\ep$ and
then passing to the limit $\ep\to 0$.

%%%%%%%%%%%%%%%%%%%%%%%%%%%%%%%%%%%%%%%%%%%%%%%%%%%%%%%%%%%%%%%%%%%%%%%%%%%%%%%%%%%%%%%%%%%%%%%%%%
%%%%%%%%%%%%%%%%%%%%%%%%%%%%%%%%%%%%%%%%%%%%%%%%%%%%%%%%%%%%%%%%%%%%%%%%%%%%%%%%%%%%%%%%%%%%%%%%%%
%%%%%%%%%%%%%%%%%%%%%%%%%%%%%%%%%%%%%%%%%%%%%%%%%%%%%%%%%%%%%%%%%%%%%%%%%%%%%%%%%%%%%%%%%%%%%%%%%%
%%%%%%%%%%%%%%%%%%%%%%%%%%%%         SIMULATION              %%%%%%%%%%%%%%%%%%%%%%%%%%%%%%%%%%%%%
%%%%%%%%%%%%%%%%%%%%%%%%%%%%%%%%%%%%%%%%%%%%%%%%%%%%%%%%%%%%%%%%%%%%%%%%%%%%%%%%%%%%%%%%%%%%%%%%%%
%%%%%%%%%%%%%%%%%%%%%%%%%%%%%%%%%%%%%%%%%%%%%%%%%%%%%%%%%%%%%%%%%%%%%%%%%%%%%%%%%%%%%%%%%%%%%%%%%%
\section{Computer simulations of the minimizer}

Consider a minimizer $j_0$ of \eqref{in5}. We say that constraint \eqref{in8} is active on a server $S_m$ if 
\begin{equation}\label{cs1}
\mu \{U:\; j_0(U)=m\}=C_m,
\end{equation}
and it is inactive on $S_m$ if
\begin{equation}\label{cs2}
\mu \{U:\; j_0(U)=m\}<C_m.
\end{equation}
If all constraints are inactive then the minimizer $j_0$ is the usual Voronoi diagram on the sphere.

In the computer simulations we begin with the usual Voronoi diagram. To that end, 
we construct the convex hull $\mathcal H$ of the set $\mathcal S$. This can be done
in $\mathcal O(n\ln n)$ operations. The convex hull gives the Delaunay triangulation,
and the spherical bisectors of the edges of the Delaunay triangulation give the edges of the Voronoi diagram
(see, e.g., \cite{bro1}, \cite{bro2}, \cite {naleeche}). We check the condition 
\begin{equation}\label{cs3}
\mu (\sg_m)\le C_m,
\end{equation}
for the Voronoi cells $\sg_m$, $m=1,\ldots, n$. If the condition is satisfied for all $m$'s, then 
the minimizer is the usual Voronoi diagram. 

%%%%%%%%%%%%% FIGURE 6  %%%%%%%%%%%%%%
\begin{figure}   [htp]
\centering
\includegraphics[width=160mm]{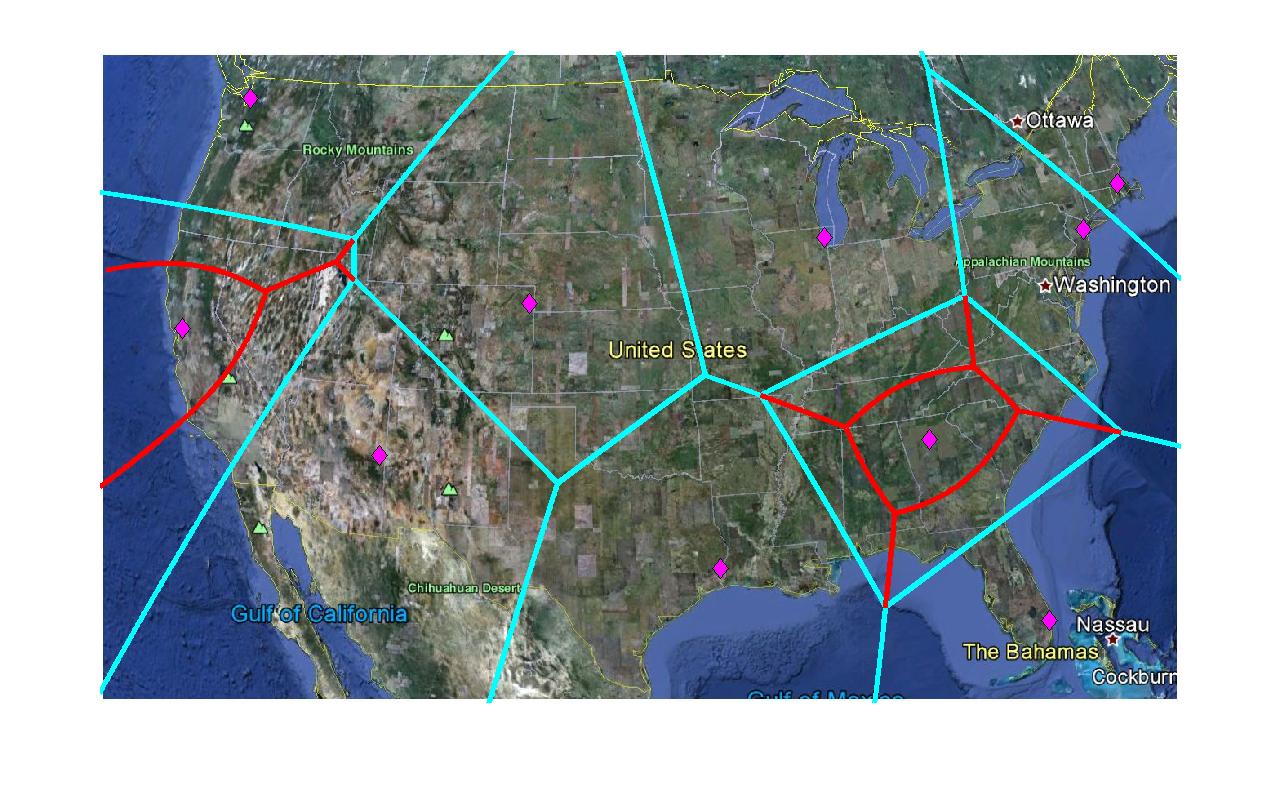}
%\vspace{-50pt}
\caption[sectors ]{The solution to the constrained minimization problem with 10 data centers, 
Seattle, Atlanta, New York, Phoenix, San Francisco, Denver, Houston, Chicago, Boston, and Miami.
The Voronoi Diagram on the sphere with these 10 points is shown by the cyan line. 
We impose two constraints shown in Red: in San Francisco the capacity 
is equal to 15 million users and in Atlanta it is 20 million.}
\label{fig:SimN}
\end{figure}
%%%%%%%%%%%%%%%%%%%%%%%%%%%%%%%%%%

In application to cloud computing, it is plausible to assume that most of the servers have
a big capacity, and condition \eqref{cs3} violates only on a set $\mathcal S_0$ of isolated $S_m$'s.
To set a notation, let
\begin{equation}\label{cs4}
\mathcal S_0=\{S_m,\; m\in I_0\},
\end{equation}
where $I_0\subset\{1,\ldots,n\}$.
In this case we are looking for parameters $\{d_m,\;m\in I_0\}$ such that  
\begin{equation}\label{cs5}
\mu (\sg_m(d))= C_m,\quad m\in I_0, 
\end{equation}
where $\sg_m(d)$ is the cell of the hyperbolic Voronoi diagram $V(d)$, with given $d_m$ when $m\in I_0$ and 
$d_m=0$ when $m\not\in I_0$. We solve equation \eqref{cs5} numerically for each $m\in I_0$. We assume that
the capacity of the neighboring vertices to $S_m\in\mathcal S_0$ is big so that the resulting
hyperbolic Voronoi diagram $V(d)$ satisfies condition \eqref{cs3} for all $m$. 

In Fig.~\ref{fig:SimN} we illustrate the numerical solution by an example. In this example, we consider 10 data centers located in
Seattle, San Francisco, Phoenix, Denver, Chicago, Houston, Atlanta, Boston, New York, and Miami. We assume that
the distribution of users is proportional to the population distribution. The blue lines show the spherical
Voronoi diagram. We further assume that the data centers in San Francisco and Atlanta have a limited capacity
equivalent to the population of 15 and 20 million, respectively. We calculate the hyperbolic Voronoi diagram 
under these constraints. The change from the usual Voronoi diagram to the hyperbolic one is shown 
in Fig.~\ref{fig:SimN} in red. 

%%%%%%%%%%%%%%%%%%%%%%%%%%%%%%%%%%%%%%%%%%%%%%%%%%%%%%%%%%%%%%%%%%%%%%%%%%%%%%%%%%%%%%%%%%%%%%%%%%
%%%%%%%%%%%%%%%%%%%%%%%%%%%%%%%%%%%%%%%%%%%%%%%%%%%%%%%%%%%%%%%%%%%%%%%%%%%%%%%%%%%%%%%%%%%%%%%%%%
%%%%%%%%%%%%%%%%%%%%%%%%%%%%%%%%%%%%%%%%%%%%%%%%%%%%%%%%%%%%%%%%%%%%%%%%%%%%%%%%%%%%%%%%%%%%%%%%%%
%%%%%%%%%%%%%%%%%%%%%%%%%%%%         END OF SIMULATION       %%%%%%%%%%%%%%%%%%%%%%%%%%%%%%%%%%%%%
%%%%%%%%%%%%%%%%%%%%%%%%%%%%%%%%%%%%%%%%%%%%%%%%%%%%%%%%%%%%%%%%%%%%%%%%%%%%%%%%%%%%%%%%%%%%%%%%%%
%%%%%%%%%%%%%%%%%%%%%%%%%%%%%%%%%%%%%%%%%%%%%%%%%%%%%%%%%%%%%%%%%%%%%%%%%%%%%%%%%%%%%%%%%%%%%%%%%%

\section{Conclusion}

In this work we studied the minimization problem for the total communication distance in a computer
cloud under the condition of restricted capacity of the data centers. 
Our main result is Theorem \ref{thm1}, which shows that a solution to the minimization problem 
is given by a hyperbolic Voronoi diagram constructed on the data centers
$S_1,\ldots,S_n$. The parameters $d_1,\ldots,d_n$ of the hyperbolic Voronoi diagram can be found from
the condition that the number of users in each cell $\sg_m$ of the diagram does not exceed the capacity
of the corresponding data center $S_m$.

Although we discuss the application to the computer cloud only, the hyperbolic Voronoi diagrams on the sphere
can find other important applications. We can mention the problem of location of air-bases \cite{OBS}, the 
distribution of facilities in global Internet companies like Amazon.com, the distribution of the 
telecommunication centers for mobile telephones in global telephone companies, data collection centers,
and others. Also, theorem \ref{thm1} can be extended, under appropriate conditions, 
to Riemannian manifolds in dimension 2 and higher.

\end{document}